\documentstyle[12pt]{article}

\newcommand{\mn}{\sl}
\def\afterthmseparator{}
\makeatletter
\renewcommand{\@begintheorem}[2]{\trivlist
      \item[\hskip \labelsep{\bf #1\ #2\unskip\afterthmseparator}]\mn}
\renewcommand{\@opargbegintheorem}[3]{\trivlist
      \item[\hskip \labelsep{\bf #1\ #2\ (#3)\unskip\afterthmseparator}]\mn}
\makeatother
\newtheorem{theorem}{Theorem}[section]
\newtheorem{lemma}{Lemma}[section]
\newtheorem{corollary}{Corollary}[section]
\newtheorem{proposition}{Proposition}[section]
\newtheorem{rem}{Remark}[section]
\newenvironment{remark}{\renewcommand{\mn}{\rm} \begin{rem}}{\end{rem}}

\newcommand{\qed}{$\;\;\;\Box$}
\newenvironment{proof}{\par\smallbreak{\sl Proof.~}}
{\unskip\nobreak\hfill \qed \par\medbreak}

\newcommand{\function}[2]{:#1 \rightarrow #2}
\newcommand{\absvalue}[1]{\left| #1 \right|}
\newcommand{\setdef}[2]{\left\{ \hspace{0.5mm} #1 :
\hspace{0.5mm} #2 \right\}}
\newcommand{\bbn}{ {\bf N} }

\newcommand{\PP}[1]{ {\bf P} \left[ #1 \right] }
\newcommand{\condP}[2]{ {\bf P} \left[ \hspace{0.5mm} #1
\left| \hspace{0.5mm} #2 \right. \right] }

\newcommand{\refeq}[1]{(\ref{#1})}

\newcommand{\dcm}{{\sc DCM}}
\newcommand{\dcnm}{{\sc DCNM}}
\newcommand{\memb}{{\sc Membership in a Permutation Group}}
\newcommand{\problem}[3]{\par\smallskip\par\noindent #1\\
{\it Given:} #2.\\ {\it Recognize if:} #3. \par\smallskip\par}
\newcommand{\view}{\mbox{view}}
\newcommand{\completeness}{\par{\it Completeness.\ }}
\newcommand{\soundness}{\par{\it Soundness.\ }}
\newcommand{\zk}{\par{\it Zero-knowledge.\ }}

\newsavebox{\biblio}
\savebox{\biblio}{\parbox{\textwidth}{\normalsize This paper was published in
{\it Algebraic structures and their applications.}
Proceedings of the third international algebraic conference held in
framework of the Ukrainian mathematical congress (Kiev, 2001),
pages 351--363. Institute of Mathematics, Ukrainian Academy of Sciences (2002).

\mbox{}\qquad
Reviewed in {\it Mathematical Review} MR2210506 (2006m:68054) and
{\it Zentralblatt f\"ur Mathematik} Zbl 1099.20501.}}

\title{
\mbox{}\\[-45mm]\usebox{\biblio}\\[5mm]
On the Double Coset Membership Problem\\
for Permutation Groups}

\author{Oleg Verbitsky\\
{\small Department of Algebra}\\
{\small Faculty of Mechanics \& Mathematics}\\
{\small Kyiv National University}\\
{\small Volodymyrska 60}\\
{\small 01033 Kyiv, Ukraine}
}

\date{}

\begin{document}

\maketitle

\begin{abstract}
\noindent
We show that the Double Coset Membership problem for permutation groups
possesses perfect zero-knowledge proofs.
\end{abstract}

\section{Introduction}

\subsection{Definition of the problem}

Let $S_m$ be a symmetric group of order $m$.
We suppose that an element of $S_m$, a permutation of an $m$-element
set, is encoded by a binary string of length $n=\lceil\log_2m!\rceil$,
$m(\log_2m-O(1))\le n\le m\log_2m$. Whenever
we refer to a {\em permutation group\/} $G$, we mean that
$G$ is a subgroup of $S_m$ for some $m$. Throughout the paper
we assume that permutation groups are given by a list of their
generators.

In this paper we address the following algorithmic problem
considered first by Luks~\cite{Luk}.
\problem{\dcm\/ ({\sc Double Coset Membership})}
{two permutations $\sigma$ and $\tau$ and two permutation groups $G$ and $H$,
all of the same order}
{$\sigma\in G\tau H$}

\subsection{Current complexity status}

For the background on computational complexity theory the reader
is referred to~\cite{GJo}.

\dcm\/ is in the class NP by the Babai-Szemer\'edy Reachability
Theorem~\cite{BSz}. This theorem says that, given any set $S$
of generators of a finite group $G$ and any $g\in G$, there exists
a sequence of elements $u_1,\ldots,u_l$ of $G$ such that the
following conditions are met.
\begin{enumerate}
\item
Each $u_i$ either belongs to $S$ or is obtained by the inversion
or the group operation from one or two previous elements of the sequence.
\item
$u_l=g$.
\item
$l\le (1+\log_2|G|)^2$.
\end{enumerate}

As $\sigma\in G\tau H$ iff $\tau^{-1}\sigma\in (\tau^{-1}G\tau)H$,
\dcm\/ admits the following reformulation.
\problem{\dcm\/ (An equivalent formulation)}
{a permutation $s$ and two permutation groups $G$ and $H$,
all of the same order}
{$s\in GH$}

Consider two related problems, the first one easier and the second one
harder than \dcm.

\problem{\memb}
{a permutation $s$ and a permutation group $G$ of the same order}
{$s\in G$}

\problem{{\sc Membership in a 3-fold Group Product}}
{a permutation $s$ and three permutation groups $G$, $H$, and $K$,
all of the same order}
{$s\in GHK$}

It is known that the former problem is solvable in polynomial time
\cite{Sim,FHL} and that the latter problem is NP-complete~\cite{Luk2}.
There are evidences that the complexity of \dcm\/ is strictly
in between. On the one hand, the problem of recognition if
two given graphs are isomorphic is polynomial-time reducible
to \dcm\/ \cite{Luk}, see also Proposition \ref{prop:gi} below.
\dcm\/ is therefore not expected to be solvable in polynomial time
as long as the Graph Isomorphism problem is not solved in polynomial
time (the currently best algorithm due to Luks and Zemplyachenko
runs in time $\exp(O(\sqrt{n\log n}))$
for graphs on $n$ vertices, see \cite{BLu}).
On the other hand, \dcm\/ belongs to the complexity class coAM
(see Subsection \ref{ss:zkdefn} for the definition).
By \cite{BHZ}, if NP is a subclass of coAM, then the polynomial-time
hierarchy of complexity classes collapses to its second level, i.e.,
$\Sigma_2^P=\Pi_2^P$ (see~\cite{GJo}). As the latter consequence
is widely considered unlikely, it is unlikely that \dcm\/ is NP-complete.

Like the membership in coAM, some other complexity-theoretic results
known for Graph Isomorphism also generalize to \dcm. Both the problems
have {\em program checkers\/} \cite{BKa}, and both are {\em low\/}
for the complexity class PP \cite{KST2}.

It is worth noting that several other group-theoretic problems are
polynomial-time equivalent with \dcm. We mention a few examples
from the list of such problems compiled in \cite{Luk,Hof2}:
Given permutation groups $G$, $H$ and permutations $\sigma$, $\tau$,
(a) find generators for $G\cap H$;
(b) recognize if $G\sigma$ and $H\tau$ intersect;
(c) if $\sigma\in G$, find the centralizer of $\sigma$ in $G$;
(d) if $\sigma,\tau\in G$, recognize if the centralizer of $\tau$
in $S_m$ intersects $G\sigma$.
In \cite{BKa} it is shown that \dcm\/ is equivalent with the problem,
given $s\in GH$, to find a factorization $s=gh$ with $g\in G$ and $h\in H$.

\subsection{Our result}

A natural question to ask about an NP problem whose polynomial-time
solvability and NP-completeness are unknown is if it possesses
a perfect or a statistical zero-knowledge interactive proof system.
Informally speaking, a zero-knowledge proof system for a recognition
problem of a language $L$ is a protocol
for two parties, the prover and the verifier, that allows the prover to
convince the verifier that a given input belongs to $L$,
with high confidence but without communicating the verifier any information
(the rigorous definitions are in Subsection~\ref{ss:zkdefn}).

The concept of a zero-knowledge proof has notable applications
in designing cryptographic protocols and in estimating the computational
complexity of a language recognition problem. Namely, by \cite{AHa}
the class PZK of languages having perfect zero-knowledge proof systems
is a subclass of coAM. Thus, the existence of a perfect zero-knowledge
proof of the membership in $L$ not only has a cryptographic meaning
but also implies that $L$ is in coAM and hence cannot be NP-complete
unless the polynomial-time hierarchy collapses.

For the Graph Isomorphism problem, its membership in coAM was proven
directly in \cite{Scho} and its membership in PZK was proven in~\cite{GMW}.
For \dcm, the proof of its membership in coAM given in \cite{BMo}
is direct. In the present paper we prove that \dcm\/ is also in PZK.
We therefore extend the list of problems in
PZK that currently includes Graph Isomorphism \cite{GMW},
Quadratic Residuosity \cite{GMR}, a problem equivalent to
Discrete Logarithm \cite{GKu}, and approximate versions of the
Shortest Vector and Closest Vector problems for integer lattices~\cite{GGo}.

\section{Background on zero-knowledge proofs}

\subsection{Definitions}\label{ss:zkdefn}

We denote the {\em length\/} of a binary word $w$ by $|w|$.
We consider {\em languages\/} over the binary alphabet which are
subsets of $\{0,1\}^*$. The {\em complement\/} of $L$ is the language
$\bar L=\{0,1\}^*\setminus L$. Note that the \dcm\/ problem can be
represented as
a recognition problem for the language $L=\setdef{(s,G,H)}{s\in GH}$,
where $(s,G,H)$ is a suitable binary encoding of the triplet consisting
of a permutation $s$ and the lists of generators for permutation
groups $G$ and $H$.

We use the standard computational model of a deterministic {\em
Turing machine}, abbreviated further on as TM. We assume that
a TM has three tapes, namely, the input tape, the output tape,
and the work tape where all computations are performed.

A {\em probabilistic\/} TM, abbreviated further on as PTM, in addition has
the fourth tape containing a potentially infinite random binary string.
Assuming that a PTM halts on input $w$ and random string $r$,
we denote its running time by $t(w,r)$. A PTM
is {\em polynomial-time\/} if $t(w,r)$ is bounded by a polynomial
in $|w|$ for all $w$ and $r$. Assuming that a PTM halts on $w$ for
almost all $r$, the function $t(w,r)$ for a fixed $w$ can be
considered as a random
variable on the probability space $\{0,1\}^\bbn$ of all random strings.
A PTM is {\em expected polynomial-time\/} on $L\subseteq\{0,1\}^*$
if for all $w\in L$ the expectation
of $t(w,r)$ is bounded by a polynomial in $|w|$.

An {\em interactive proof system\/} $\langle V,P\rangle$, further on
abbreviated as IPS, consists of
two PTMs, a polynomial-time $V$ called {\em the verifier\/} and
a computationally unlimited $P$ called {\em the prover}.
The input tape is common for the verifier and the prover.
The verifier and the prover also share a communication tape
which allows message exchange between them.
The system works as follows. First both the machines $V$ and $P$
are given an input $w$ and each of them is given an individual
random string, $r_V$ for $V$ and $r_P$ for $P$. Then $P$ and $V$
alternatingly write messages to one another in the communication tape.
$V$ computes its $i$-th message $a_i$ to $P$ based on the input $w$,
the random string $r_V$, and all previous messages from $P$ to $V$.
$P$ computes its $i$-th message $b_i$ to $V$ based on the input $w$,
the random string $r_P$, and all previous messages from $V$ to $P$.
After a number of message exchanges $V$ terminates interaction and
computes an output based on $w$, $r_V$, and all $b_i$. The output
is denoted by $\langle V,P\rangle(w)$.
Note that, for a fixed $w$, $\langle V,P\rangle(w)$ is a random variable
depending on both random strings $r_V$ and $r_P$.

Let $\epsilon(n)$ be a function of a natural argument taking on
positive real values.
We call $\epsilon(n)$ {\em negligible\/} if $\epsilon(n)<n^{-c}$
for every $c$ and all $n$ starting from some $n_0(c)$.
For example, an {\em exponentially small\/} function $\epsilon(n)=d^{-n}$,
where $d>1$, is negligible.

We say that $\langle V,P\rangle$ is an {\em IPS
for a language $L$ with error $\epsilon(n)$\/} if the following two
conditions are fulfilled.

\smallskip

\noindent
{\em Completeness.} If $w\in L$, then $\langle V,P\rangle(w)=1$
with probability at least $1-\epsilon(|w|)$.\\
{\em Soundness.} If $w\notin L$, then, for an arbitrary interacting
PTM $P^*$, $\langle V,P^*\rangle(w)=1$ with probability at most
$\epsilon(|w|)$.

\smallskip

\noindent
We will call any prover $P^*$ interacting with $P$ on input
$w\notin L$ {\em cheating}.
If in the completeness condition we have $\langle V,P\rangle(w)=1$
with probability 1, we say that $\langle V,P\rangle$ has {\em one-sided
error\/} $\epsilon(n)$.

We say that $\langle V,P\rangle$ is an {\em IPS
for a language $L$\/} if $\langle V,P\rangle$ is an IPS
for $L$ with negligible error.

An IPS is {\em public-coin\/} if the concatenation
$a_1\ldots a_k$ of the verifier's messages is a prefix of his random
string $r_V$. A {\em round\/} is sending one message from the verifier
to the prover or from the prover to the verifier.
The class AM consists of those languages having
IPSs with error $1/3$ and with number of rounds bounded by a constant
for all inputs. A language $L$ belongs to the class coAM iff its
complement $\bar L$ belongs to AM.

Given an IPS $\langle V,P\rangle$ and an input $w$,
let $\view_{V,P}(w)=(r'_V,a_1,b_1,\ldots,a_k,b_k)$ where $r'_V$ is a part
of $r_V$ scanned by $V$ during work on $w$ and $a_1,b_1,\ldots,a_k,b_k$ are
all messages from $P$ to $V$ and from $V$ to $P$
($a_1$ may be empty if the first message
is sent by $P$). Note that the verifier's messages $a_1,\ldots,a_k$
could be excluded because they are efficiently computable from the other
components. For a fixed $w$, $\view_{V,P}(w)$ is a random variable depending
on $r_V$ and $r_P$.

An IPS $\langle V,P\rangle$ is {\em perfect
zero-knowledge on $L$\/} if for every interacting polynomial-time PTM $V^*$
there is a PTM $M_{V^*}$, called a {\em simulator},
that on every input $w\in L$ runs in expected polynomial time
and produces output $M_{V^*}(w)$ which,
if considered as a random variable depending on a random string
of $M_{V^*}$, is distributed identically with $\view_{V^*,P}(w)$.
The latter condition means that
$$
\PP{M_{V^*}(w)=z}=\PP{\view_{V^*,P}(w)=z}\mbox{ for all }z.
$$
If only a weaker condition that
$$
\sum_z\absvalue{\PP{M_{V^*}(w)=z}-\PP{\view_{V^*,P}(w)=z}}
\mbox{ is negligible}
$$
is true, we call $\langle V,P\rangle$ {\em statistical zero-knowledge}.
These notions formalize the claim that the verifier gets no information
during interaction with the prover: Everything that the verifier gets
he can get without the prover by running the simulator.

According to the definition the verifier learns nothing even if he
deviates from the original program and follows an arbitrary probabilistic
polynomial-time program $V^*$. We will call the verifier $V$ {\em honest\/}
and all other verifiers $V^*$ {\em cheating}. If the existence of
a simulator is claimed
only for the honest verifier, we call such a proof system
{\em honest-verifier perfect ({\em or\/} statistical) zero-knowledge}.

The class of languages $L$ having IPSs that are perfect
(resp.\ statistical) zero-knowledge on $L$
is denoted by PZK (resp.\ SZK). Recall that the error
here is supposed negligible.

The {\em $k(n)$-fold sequential composition\/} of an
IPS $\langle V,P\rangle$ is the IPS
$\langle V',P'\rangle$ in which $V'$ and $P'$ on input $w$ execute
the programs of $V$ and $P$ sequentially $k(|w|)$ times, each time
with independent choice of random strings $r_V$ and $r_P$.
At the end of interaction $V'$ outputs 1 iff $\langle V,P\rangle(w)=1$
in all $k(|w|)$ executions. The initial system $\langle V,P\rangle$
is called {\em atomic}.

In the {\em $k(n)$-fold parallel composition\/} $\langle V'',P''\rangle$
of $\langle V,P\rangle$, the program of $\langle V,P\rangle$
is executed $k(|w|)$ times in parallel, that is, in each round
all $k(|w|)$ versions of a message are sent from one machine to another
at once as a long single message. In every parallel execution
$V''$ and $P''$ use independent copies of $r_V$ and $r_P$.
At the end of interaction $V'$ outputs 1 iff $\langle V,P\rangle(w)=1$
in all $k(|w|)$ executions.

\subsection{Known results on zero-knowledge proofs}

We first notice a simple property of sequential composition of IPSs.

\begin{proposition}\label{prop:seqrep}
If $\langle V,P\rangle$ is an IPS for a language $L$ with one-sided
constant error $\epsilon$, then the $k(n)$-fold sequential composition
of $\langle V,P\rangle$ is an IPS for $L$ with one-sided error
$\epsilon^{k(n)}$.
\end{proposition}

Parallel composition obviously preserves the number of rounds,
the public-coin property, and the property of error to be one-sided.
It is not hard to prove that $k$-fold parallel composition reduces
the one-sided error $\epsilon$ to $\epsilon^k$.
It is also not hard to prove that parallel composition preserves
perfect and statistical zero-knowledge for the honest verifier.
These observations are summarized in the next proposition.

\begin{proposition}\label{prop:parrep}
Assume that $\langle V,P\rangle$ is a honest-verifier perfect zero-knowledge
public-coin IPS for a language $L$ that on all inputs works in
a constant $c$ rounds with one-sided constant error $\epsilon$.
Then $k(n)$-fold parallel composition of $\langle V,P\rangle$
is a honest-verifier perfect zero-knowledge IPS for $L$
that works in $c$ rounds with error $\epsilon^{k(n)}$.
\end{proposition}

We also refer to the following deep results in the theory of
zero-knowledge proofs.

\begin{proposition}[Aiello-H\aa stad \cite{AHa}]\label{prop:aha}
$\mbox{SZK}\subseteq\mbox{coAM}$.
\end{proposition}

\begin{proposition}[Okamoto \cite{Oka}]\label{prop:oka}
\mbox{}

\begin{enumerate}
\item
Every honest-verifier statistical zero-knowledge IPS for
a language $L$ can be transformed in an honest-verifier statistical
zero-knowledge public-coin IPS for~$L$.
\item
If $L$ has an honest-verifier statistical zero-knowledge public-coin IPS,
then $\bar L$ has a honest-verifier statistical zero-knowledge
constant-round IPS.
\end{enumerate}
\end{proposition}
Note that the item 2 of Proposition \ref{prop:oka} strengthens
Proposition \ref{prop:aha} because by \cite{GSi} every IPS can be made
public-coin at cost of decreasing the number of rounds in 2.

\begin{proposition}[Goldreich-Sahai-Vadhan \cite{GSV}]\label{prop:gsv}
Every honest-verifier statistical zero-knowledge public-coin IPS for
a language $L$ can be transformed in a general
statistical zero-knowledge public-coin IPS for $L$.
If the error of the initial IPS is one-sided, so is the error
of the resulting IPS.
\end{proposition}

Note that, to achieve the negligible error, the transformation
of Proposition \ref{prop:gsv} makes the number of rounds increasing
with the input size increasing, even if the initial IPS is constant-round.
A transformation preserving the constant number of rounds is known
only under an unproven assumption about the hardness of the
Discrete Logarithm problem (the formal statement of the assumption
can be found in~\cite{BMOs}).

\begin{proposition}[Bellare-Micali-Ostrovsky \cite{BMOs}]\label{prop:bmo}
Suppose that a language $L$ has an honest-verifier statistical zero-knowledge
IPS that on every input $w$ works in $c(|w|)$ rounds with error at most
$1/3$. Then, under the assumption on the hardness of Discrete Logarithm,
$L$ has a general statistical zero-knowledge IPS that on input $w$ works
in $O(c(|w|))$ rounds with exponentially small error.
\end{proposition}

\section{Background on permutation groups}

Given a finite set $X$, by a {\em random element\/} of $X$ we mean
a random variable uniformly distributed over $X$.

\begin{proposition}[Sims \cite{Sim,FHL}]\label{prop:sim}
\mbox{}

\begin{enumerate}
\item
There is a polynomial-time algorithm for recognizing the \memb.
\item
There is a probabilistic polynomial-time algorithm that, given
a list of generators for a permutation group $G$, outputs a random
element of $G$.
\end{enumerate}
\end{proposition}

The \dcm\/ problem is at least as hard as testing isomorphism of two
given graphs.
\begin{proposition}[Luks \cite{Luk}, Hoffmann \cite{Hof1}]\label{prop:gi}
The Graph Isomorphism problem is polynomial-time reducible to \dcm.
\end{proposition}
We include a proof for the sake of completeness.
\begin{proof}
Consider two graphs of order $n$ with adjacency matrices $A=(a_{ij})$
and $B=(b_{ij})$. Let $S_1=\setdef{(i,j)}{a_{ij}=1}$ and
$S_2=\setdef{(i,j)}{b_{ij}=1}$.

Let $G$ be the group of permutations
of the square $\{1,\ldots,n\}^2$ generated by simultaneous transpositions
of $i$-th and $j$-th rows and $i$-th and $j$-th columns for all
$1\le i<j\le n$. The graphs are isomorphic iff $G$ contains a permutation
$\sigma$ such that $\sigma(S_1)=S_2$.

Let $H$ be the group of permutations $\tau$ such that $\tau(S_1)=S_1$
and $s$ be an arbitrary permutation such that $s(S_1)=S_2$.
As easily seen, a permutation $\sigma$ as above exists iff $s\in GH$.
\end{proof}
Note that the reduction described allows one to transform any zero-knowledge
proof system for \dcm\/ in a zero-knowledge proof system for Graph
Isomorphism.

\section{Zero-knowledge proofs for \dcm}

\begin{theorem}\label{thm:1}
The \dcm\/ problem has an honest-verifier perfect zero-knowledge
three-round public-coin IPS with one-sided error $1/2$.
\end{theorem}

\begin{proof}
On input $(s,G,H)$ such that $s\in GH$ the IPS $\langle V,P\rangle$
proceeds as follows.

\smallskip

{\it 1st round.}

\noindent
$P$ generates random elements $g\in G$ and $h\in H$, computes $t=gsh$,
and sends $t$ to $V$. $V$ checks if $t$ is a permutation of the given
order and if not (this is possible in the case of a cheating prover)
halts and outputs 1.

\medskip

{\it 2nd round.}

\noindent
$V$ chooses a random bit $b\in\{0,1\}$ and sends it to $P$.

\smallskip

{\it 3rd round.}

\noindent
{\it Case $b=0$.}
$P$ sends $V$ permutations $g$ and $h$. $V$ checks if
$g\in G$, $h\in H$, and $t=gsh$.

\noindent
{\it Case $b\ne 0$} (this includes the possibility of a message
$b\notin\{0,1\}$ produced by a cheating verifier).
$P$ decomposes $s$ into the product $s=g_0h_0$ with $g_0\in G$
and $h_0\in H$, computes $g_1=gg_0$ and $h_1=h_0h$, and sends $g_1$
and $h_1$ to $V$. $V$ checks if $g_1\in G$, $h_1\in H$, and $t=g_1h_1$.

$V$ halts and outputs 1 if the conditions are checked successfully
and 0 otherwise.

\medskip

This IPS is obviously public-coin. We need to check that this is indeed
an IPS for \dcm\/ with one-sided error 1/2 and, moreover, that this
is a honest-verifier perfect zero-knowledge IPS.

\completeness
If $s\in GH$, then it is clear that $V$ outputs 1 with probability 1.

\soundness
Assume that $s\notin GH$ and consider an arbitrary cheating prover $P^*$.
Observe that if both $t=gsh$, $g\in G$, $h\in H$ and $t=g_1h_1$,
$g_1\in G$, $h_1\in H$, then $s\in GH$. It follows that,
for at least one value of $b$, $V$ outputs 0 and therefore $V$
outputs 1 with probability at most 1/2.

\zk
Assume that $s\in GH$. During interaction with $P$, $V$ sees
$\view_{V,P}(s,G,H)=(b,t,b,g',h')$ where $g'$ and $h'$ are
received by $V$ in the 3rd round. If $b=0$, then $t=gsh$, $g'=g$,
and $h'=h$. If $b=1$, then $t=g'h'$, $g'=gg_0$, and $h'=h_0h$.
In both the cases $g'$ and $h'$ are random elements of $G$ and $H$
respectively. The random variable $\view_{V,P}(s,G,H)$ can be therefore
generated by the following simulator: Generate a random bit $b$
and random elements $g'\in G$ and $h'\in H$; If $b=0$, set $t=g'sh'$;
If $b=1$, set $t=g'h'$.
\end{proof}

\begin{corollary}\label{cor:1}
The \dcm\/ problem has an honest-verifier perfect zero-knowledge
three-round public-coin IPS with one-sided error $2^{-n}$.
\end{corollary}

\begin{proof}
By Proposition \ref{prop:parrep} the $n$-fold parallel composition
of the IPS from Theorem \ref{thm:1} reduces the error to $2^{-n}$
and preserves the properties of the atomic system.
\end{proof}

Let {\sc Double Coset Non-Membership}, abbreviated as \dcnm,
be the problem opposite to \dcm, that is, given a permutation $s$
and two permutation groups $G$ and $H$, to recognize if $s\notin GH$.
The \dcnm\/ problem is clearly polynomial-time equivalent with
recognition of the set-theoretic complement of \dcm, where
the latter is encoded as a language in the binary alphabet.

\begin{corollary}\label{cor:2}
\dcnm\/ has an honest-verifier statistical zero-knowledge constant-round
IPS.
\end{corollary}

\begin{proof}
The corollary follows from Corollary \ref{cor:1} by Proposition
\ref{prop:oka}.

We also give an alternative direct proof of this claim
describing an honest-verifier perfect zero-knowledge two-round IPS
$\langle V,P\rangle$
for \dcnm\/ with one-sided error 1/2. This system, for the case of
permutation groups, generalizes the IPS suggested in \cite{Bab3}
for the problem of testing the membership in a finite group given
by a list of generators and an oracle access to the group operation.

On input $(s,G,H)$ such that $s\notin GH$ the system works as follows.

\medskip

{\it 1st round.}

\noindent
$V$ chooses a random bit $b$ to be the first bit of a random string $r_V$
and, based on the subsequent bits of $r_V$, generates random elements
$g\in G$ and $h\in H$.
If $b=0$, $V$ computes $t=gh$; If $b=1$, $V$ computes $t=gsh$.
Then $V$ sends $t$ to $P$.

{\it 2nd round.}

\noindent
$P$ recognizes if $t\in GH$. If so, $P$ sets $a=0$; If not, $P$ sets $a=1$.
Then $P$ sends $a$ to $V$.

$V$ checks if $a=b$ and halts. If the equality is true, $V$ outputs 1;
Otherwise $V$ outputs 0.

\medskip

\completeness
Assume that $s\notin GH$. In the first round, $t\in GH$ if $b=0$
and $t\notin GH$ if $b=1$. Therefore $V$ outputs 1 with probability 1.

\soundness
Assume that $s\in GH$. Then $t\in GH$ regardless of the value of $b$.
Moreover, $t$ is the product of random elements of $G$ and $H$
and, as a random variable, is independent of the random variable $b$.
It follows that in the second round a message from the cheating prover
$P^*$ to $V$, which is a function of $s$, $G$, $H$, $r_{P^*}$, and $t$,
is equal to $b$ with probability at most 1/2. Hence
$\langle V,P^*\rangle(s,G,H)=1$ with probability at most 1/2.

\zk
Assume that $s\notin GH$. During interaction with $P$,
$V$ sees
$\view_{V,P}(s,G,H)=(r'_V,t,a)$,
where $a$ equals the first bit $b$
of $r'_V$. The simulator therefore just generates a random string $r_V$,
extracts the first bit $b$ from it, sets $a=b$, based on the remaining bits
of $r_V$ computes $g$ and $h$, based on $b$, $g$, and $h$ computes $t$,
and sets $r'_V$ to be the prefix of $r_V$
that was actually used for these purposes.
\end{proof}

\begin{corollary}\label{cor:3}
\dcm\/ is in SZK. Moreover, \dcm\/ has a statistical zero-knowledge
public-coin IPS with one-sided error.
\end{corollary}

\begin{proof}
Apply the transformation from Proposition \ref{prop:gsv} to the IPS
from Corollary~\ref{cor:1}.

Note that another proof of the membership of \dcm\/ in SZK
can be given by applying Propositions \ref{prop:oka} and
\ref{prop:gsv} to the IPS in the alternative proof of
Corollary~\ref{cor:2}.
\end{proof}

\begin{corollary}[Babai-Moran \cite{BMo}]\label{cor:4}
\dcm\/ is in coAM. Therefore \dcm\/ is not NP-complete unless
the polynomial-time hierarchy collapses at the second level.
\end{corollary}

\begin{proof}
This is an immediate consequence of Corollary \ref{cor:2}
or a consequence of Corollary \ref{cor:3} based on
Proposition~\ref{prop:aha}.
\end{proof}

\begin{corollary}\label{cor:5}
Under the assumption on the hardness of Discrete Logarithm,
\dcm\/ has a constant-round statistical zero-knowledge IPS
with exponentially small error.
\end{corollary}

\begin{proof}
The corollary follows from Theorem \ref{thm:1} by Proposition~\ref{prop:bmo}.
\end{proof}

\begin{theorem}\label{thm:2}
The $n$-fold sequential composition of the IPS in Theorem \ref{thm:1}
is a perfect zero-knowledge public-coin IPS for \dcm\/ with
exponentially small error. Hence \dcm\/ is in PZK.
\end{theorem}

\begin{proof}
Denote the composed IPS by $\langle V,P\rangle$.
As the atomic system is public-coin, so is $\langle V,P\rangle$.
By Proposition \ref{prop:seqrep} $\langle V,P\rangle$ is an IPS
for \dcm\/ with one-sided error $2^{-n}$. We have to prove that
$\langle V,P\rangle$ is perfect zero-knowledge.

For each verifier $V^*$ interacting with $P$ we describe a probabilistic
expected polynomial-time simulator $M_{V^*}$. The $M_{V^*}$ uses
the program of $V^*$ as a subroutine. Assume that the running time
of $V^*$ is bounded by a polynomial $q(n)$ in the input size.
On input $w$, $M_{V^*}$ will run the program of $V^*$ on input $w$
with random string $r$, where $r$ is the prefix of $M_{V^*}$'s random
string of length $q(|w|)$. In all other cases $M_{V^*}$ will use
the remaining part of its random string.

Work of $M_{V^*}$ on input $w=(s,G,H)$ consists of $|w|$ stages, where
a stage corresponds to an iteration of the atomic system.

\smallskip

{\it Stage $i$.}

\smallskip

\noindent
$M_{V^*}$ chooses random elements $g_i\in G$ and $h_i\in H$ and a
random bit $a\in\{0,1\}$. If $a=0$, $M_{V^*}$ computes $t_i=g_ish_i$;
If $a=1$, it computes $t_i=g_ih_i$. Then $M_{V^*}$ computes
$b_i=V^*(w,r,t_1,g_1,h_1,\ldots,t_{i-1},g_{i-1},h_{i-1},t_i)$,
the message that $V^*(w,r)$ sends $P$ in the $i$-th sequential
iteration of the atomic system after receiving $P$'s message $t_i$
and under the condition that in the preceding iterations
$P$'s messages were $t_1,g_1,h_1,\ldots,t_{i-1},g_{i-1},h_{i-1}$.
If $b_i$ and $a$ are simultaneously equal to or different from 0,
then $M_{V^*}$ puts $v_i=(t_i,b_i,g_i,h_i)$ and proceeds to the $(i+1)$-th
stage. If exactly one of $b_i$ and $a$ is equal to 0, then $M_{V^*}$
restarts the same $i$-th stage with new independent choice of $a$, $g_i$,
$h_i$.

\smallskip

After all stages are completed, $M_{V^*}$ halts and outputs
$(r',v_1,\ldots,v_{|w|})$, where $r'$ is the prefix of $r$ actually used
by $V^*$ during interaction on input $w$ with the prover sending
the messages $t_1,g_1,h_1,\ldots,t_{|w|},g_{|w|},h_{|w|}$.
Notice that it might happen that in unsuccessful attempts to pass
some stage $V^*$ used a prefix of $r$ longer than $r'$.

\medskip

We first check that $M_{V^*}$ terminates in expected polynomial time
whenever $s\in GH$. Since $V^*$ is polynomial-time, one attempt to
pass Stage $i$, $i\le|w|$, takes time bounded by a polynomial in $|w|$.
Recall that $M_{V^*}$ is programmed so that $a$ and $r$ are independent.
Furthermore, $a$ and $t_i$ are independent. Indeed, if $a=1$, then
$t_i=g_ih_i$ is the product of random elements of $G$ and $H$.
If $a=0$, then $t_i=(g_ig_0)(h_0h_i)$ is such a product as well.
Here $g_0\in G$ and $h_0\in H$ are elements of an arbitrary decomposition
$s=g_0h_0$. It follows that $a$ and $b_i$ are independent and therefore
an execution of the stage is successful with probability 1/2.
We conclude that on average each stage consists of 2 executions.
Thus, on average $M_{V^*}$ makes $2|w|$ polynomial-time executions
and this takes expected polynomial time.

We finally need to check that, whenever $s\in GH$, the output
$M_{V^*}(w)$ is distributed identically with $\view_{V^*,P}(w)$.
Notice that both the random variables depend on $V^*$'s random string $r$.
It therefore suffices to show that the distributions are identical
when conditioned on an arbitrary fixed $r$. For $0\le i\le|w|$,
let $D^i_M(w,r)$ denote the probability distribution of
$(r',v_1,\ldots,v_i)$ conditioned on $r$, and $D^i_{V^*,P}(w,r)$
denote the distribution of the part of $\view_{V^*,P}(w)$
formed up to the $i$-th sequential iteration. With this notation,
we have to prove that $D^{|w|}_M(w,r)=D^{|w|}_{V^*,P}(w,r)$.
Using the induction on $i$, we prove that
$D^{i}_M(w,r)=D^{i}_{V^*,P}(w,r)$ for every $0\le i\le|w|$.

The base case of $i=0$ is trivial. Let $i\ge 1$ and assume that
\begin{equation}\label{eq:1}
\PP{D^{i-1}_M(w,r)=u_{i-1}}=\PP{D^{i-1}_{V^*,P}(w,r)=u_{i-1}}
\end{equation}
for every value $u_{i-1}$.
Given $u_{i-1}$, assume now that both
$D^{i-1}_M(w,r)=u_{i-1}$ and $D^{i-1}_{V^*,P}(w,r)=u_{i-1}$,
and under these conditions consider how the $i$-th components
$v_i=(t_i,b_i,g_i,h_i)$ are distributed in $u_i=u_{i-1}v_i$
according to $D^{i}_M(w,r)$ and $D^{i}_{V^*,P}(w,r)$.
We will show that
\begin{equation}\label{eq:2}
\parbox{.9\textwidth}{
\begin{eqnarray*}
\condP{D^{i}_M(w,r)=u_{i-1}v_i}{D^{i-1}_M(w,r)=u_{i-1}}\hspace{53mm}&&\\
=\condP{D^{i}_{V^*,P}(w,r)=u_{i-1}v_i}{D^{i-1}_{V^*,P}(w,r)=u_{i-1}}&&
\end{eqnarray*}
}
\end{equation}
for every value $v_i$.
Together with \refeq{eq:1} this will imply the identity of
$D^{i}_M(w,r)$ and $D^{i}_{V^*,P}(w,r)$.

To prove \refeq{eq:2}, we will show that according to the both
conditional distributions $v_i$ is uniformly distributed on the set
\begin{eqnarray*}
S=\Bigl\{ \hspace{0.5mm} (t,b,g,h) & : &
t\in GH,\ b=V^*(w,r,u_{i-1},t),\ g\in G,\ h\in H,\\
&& t=gsh\mbox{ if }b=0\mbox{ and }t=gh\mbox{ if }b\ne 0\Bigr. \Bigr\}.
\end{eqnarray*}

Given $t$ and $s$, define sets
$R(t)=\setdef{(g,h)}{g\in G,\, h\in H,\, gh=t}$
and $R_s(t)=\setdef{(g,h)}{g\in G,\, h\in H,\, gsh=t}$.
The first claim of the following lemma appeared in~\cite{Hof1}.

\begin{lemma}\label{lem:}
Let $k=|G\cap H|$. Assume that $s=g_0h_0$ with $g_0\in G$ and $h_0\in H$.
Then the following statements are true.

\begin{enumerate}
\item
Every $t\in GH$ has $k$ representations $t=gh$ with $g\in G$ and $h\in H$,
i.e., $|R(t)|=k$. If $t=g_1h_1$, then all other representations are
\begin{equation}\label{eq:f}
t=(g_1f)(f^{-1}h_1),
\end{equation}
where $f$ ranges over group $G\cap H$.
\item
For every $t$,
the mapping $\alpha(g,h)=(gg_0,h_0h)$ is one-to-one from $R_s(t)$ to $R(t)$.
\item
Every $t\in GH$ has $k$ representations $t=gsh$ with $g\in G$ and $h\in H$,
i.e., $|R_s(t)|=k$.
\item
If $\phi\function{G\times H}{GH}$ is defined by $\phi(g,h)=gh$,
then $|\phi^{-1}(t)|=k$ for every $t\in GH$.
\item
If $\psi\function{G\times H}{GH}$ is defined by $\psi(g,h)=gsh$,
then $|\psi^{-1}(t)|=k$ for every $t\in GH$.
\item
If $t=gh$ is the product of uniformly distributed
random elements $g\in G$ and $h\in H$,
then $t$ is uniformly distributed on $GH$.
\item
If a uniformly distributed random pair $(g,h)\in G\times H$ is
conditioned on $gh=t$ for an arbitrary fixed $t\in GH$, then
$(g,h)$ is uniformly distributed on $R(t)$.
\item
If $t=gsh$ and $g\in G$ and $h\in H$ are uniformly distributed
random elements, then $t$ is uniformly distributed on $GH$.
\item
If a uniformly distributed random pair $(g,h)\in G\times H$ is
conditioned on $gsh=t$ for an arbitrary fixed $t\in GH$, then
$(g,h)$ is uniformly distributed on $R_s(t)$.
\end{enumerate}
\end{lemma}

\begin{proof}
We first prove Item 1. Let $e$ denote the identity permutation.
Clearly that we have at least $k$ representations
of the form \refeq{eq:f}. On the other hand, every representation
$t=gh$ is of this form. Indeed, we have $(g^{-1}g_1)(h_1h^{-1})=e$
and hence both $g^{-1}g_1$ and $h_1h^{-1}$ are simultaneously in $G$
and in $H$.

To prove Item 2, observe that $\alpha$ is indeed from $R_s(t)$ to $R(t)$.
The map $\alpha'(g,h)=(gg_0^{-1},h_0^{-1}h)$ is easily seen to be
from $R(t)$ to $R_s(t)$ and inverse to $\alpha$.

Items 1 and 2 imply Item 3, Item 3 implies Item 5, and Item 5 implies
Item 8. Item 1 implies Item 4, and Item 4 implies Item 6.
Items 7 and 9 are true by the definition of $R(t)$ and $R_s(t)$.
\end{proof}

The distribution $D^i_{V^*,P}(w,r)$ conditioned on
$D^{i-1}_{V^*,P}(w,r)=u_{i-1}$ is samplable as follows.
Choose random elements $g\in G$ and $h\in H$. Compute
$t_i=gsh$ and $b_i=V^*(w,r,u_{i-1},t_i)$. If $b_i=0$, set
$g_i=g$ and $h_i=h$, otherwise set $g_i=gg_0$ and $h_i=h_0h$.
Clearly, this distribution of $(t_i,b_i,g_i,h_i)$ is over $S$.

By Item 8 of Lemma \ref{lem:}, $t_i$ is uniformly distributed on
$GH$. If $b_i=0$, then by Item 9 of Lemma \ref{lem:}, for every fixed
$t_i$, the pair $(g_i,h_i)$ is uniformly distributed on $R_s(t)$.
If $b_i\ne 0$, then by Item 2 of Lemma \ref{lem:}, for every fixed
$t_i$, the pair $(g_i,h_i)$ is uniformly distributed on $R(t)$.
It follows that $D^i_{V^*,P}(w,r)$ conditioned on
$D^{i-1}_{V^*,P}(w,r)=u_{i-1}$ is uniform on $S$.

Consider now the sampling procedure for
the distribution $D^i_{M}(w,r)$ conditioned on
$D^{i-1}_{M}(w,r)=u_{i-1}$ as in the description of the simulator
$M_{V^*}$. Under the condition that $a=0$, by Items 8 and 9 of
Lemma \ref{lem:}, $t_i$ is distributed uniformly over $GH$
and for every fixed value of $t_i$, the pair $(g_i,h_i)$ is uniformly
distributed over $R_s(t)$.
Under the condition that $a=1$, by Items 6 and 7 of
Lemma \ref{lem:}, $t_i$ is distributed uniformly over $GH$
and for every fixed value of $t_i$, the pair $(g_i,h_i)$ is uniformly
distributed over $R(t)$.
This leads to an equivalent sampling procedure:
Choose a random $t_i\in GH$, compute $b_i=V^*(w,r,u_{i-1},t_i)$;
If $b_i=0$, choose a random pair $(g_i,h_i)$ in $R_s(t_i)$,
otherwise in $R(t)$.
It follows that $D^i_{M}(w,r)$ conditioned on
$D^{i-1}_{M}(w,r)=u_{i-1}$ is uniform on $S$.
\end{proof}

\begin{remark}
The simulator in the proof of Theorem \ref{thm:2} is {\em black-box},
that is, for each $V^*$ it follows the same program that uses
the strategy of $V^*$ as a subroutine. It should be noted that
by \cite{GKr} the parallel composition of the IPS in Theorem \ref{thm:1}
is {\em not\/} zero-knowledge with black-box simulator unless
\dcm\/ is decidable in probabilistic polynomial time.
\end{remark}

\section{Future work}

A natural question arises if our results can be extended to {\em matrix
groups\/} over finite fields. One of the reasons why this case
is more complicated is that, unlike permutation groups, no efficient test
of membership for matrix groups is known. We intend to tackle this question
in a subsequent paper.


\begin{thebibliography}{10}

\bibitem{AHa}
B.~Aiello and J.~H\aa stad.
\newblock Perfect zero-knowledge languages can be recognized in two rounds.
\newblock In {\em Proc. of the {\rm 28}th IEEE Ann. Symp. on Foundations of
  Computer Science (FOCS)}, pages 439--448, 1987.

\bibitem{Bab3}
L.~Babai.
\newblock Local expansion of vertex-transitive graphs and random generation in
  finite groups.
\newblock In {\em Proc. of the {\rm 23}rd ACM Ann. Symp. on Theory of Computing
  (STOC)}, pages 164--174, 1991.

\bibitem{BLu}
L.~Babai and E.M.Luks.
\newblock Canonical labeling of graphs.
\newblock In {\em Proc. of the {\rm 15}th ACM Ann. Symp. on the Theory of
  Computing (STOC)}, pages 171--183, 1983.

\bibitem{BMo}
L.~Babai and S.~Moran.
\newblock Arthur-Merlin games: a randomized proof system, and a hierarchy of
  complexity classes.
\newblock {\em Journal of Computer and System Sciences}, 36:254--276, 1988.

\bibitem{BSz}
L.~Babai and E.~Szemer\'edi.
\newblock On the complexity of matrix group problems.
\newblock In {\em Proc. of the {\rm 25}th IEEE Ann. Symp. on Foundations of
  Computer Science (FOCS)}, pages 229--240, 1984.

\bibitem{BMOs}
M.~Bellare, S.~Micali, and R.~Ostrovsky.
\newblock The (true) complexity of statistical zero knowledge.
\newblock In {\em Proc. of the {\rm 22}nd ACM Ann. Symp. on Theory of Computing
  (STOC)}, pages 494--502, 1990.

\bibitem{BKa}
M.~Blum and S.~Kannan.
\newblock Designing programs that check their work.
\newblock {\em J. Assoc. Comput. Mach.}, 42(1):269--291, 1995.

\bibitem{BHZ}
R.~B.~Boppana, J.~H\aa stad, and S.~Zachos.
\newblock Does co-NP have short interactive proofs?
\newblock {\em Information Processing Letters}, 25:127--132, 1987.

\bibitem{FHL}
M.~L.~Furst, J.~Hopcroft, and E.~M.~Luks.
\newblock Polynomial-time algorithms for permutation groups.
\newblock In {\em Proc. of the {\rm 21}st IEEE Ann. Symp. on Foundations of
  Computer Science (FOCS)}, pages 36--41, 1980.

\bibitem{GJo}
M.~R.~Garey and D.~S.~Johnson.
\newblock {\em Computers and Intractability. {A} guide to the theory of
  {$NP$}-completeness}.
\newblock W.~H.~Freeman, 1979 (a Russian translation available).

\bibitem{GGo}
O.~Goldreich and S.~Goldwasser.
\newblock On the limits on the non-approximability of lattice problems.
\newblock In {\em Proc. of the {\rm 30}th ACM Ann. Symp. on Theory of Computing
  (STOC)}, pages 1--9, 1998.

\bibitem{GKr}
O.~Goldreich and H.~Krawczyk.
\newblock On the composition of zero-knowledge proof systems.
\newblock {\em SIAM Journal on Computing}, 25(1):169--192, 1996.

\bibitem{GKu}
O.~Goldreich and E.~Kushilevitz.
\newblock A perfect zero-knowledge proof for a decision problem equivalent to
  {D}iscrete {L}ogarithm.
\newblock {\em Journal of Cryptology}, 6:97--116, 1993.

\bibitem{GMW}
O.~Goldreich, S.~Micali, and A.~Wigderson.
\newblock Proofs that yield nothing but their validity or all languages in NP
  have zero-knowledge proof systems.
\newblock {\em J. Assoc. Comput. Mach.}, 38(3):691--729, 1991.

\bibitem{GSV}
O.~Goldreich, A.~Sahai, and S.~Vadhan.
\newblock Honest-verifier statistical zero-knowledge equals general statistical
  zero-knowledge.
\newblock In {\em Proc. of the {\rm 30}th ACM Ann. Symp. on Theory of Computing
  (STOC)}, pages 399--408, 1998.

\bibitem{GMR}
S.~Goldwasser, S.~Micali, and C.~Rackoff.
\newblock The knowledge complexity of interactive proof systems.
\newblock {\em SIAM Journal on Computing}, 18(1):186--208, 1989.

\bibitem{GSi}
S.~Goldwasser and M.~Sipser.
\newblock Private coins versus public coins in interactive proof systems.
\newblock In {\em Proc. of the {\rm 18}th ACM Ann. Symp. on the Theory of
  Computing (STOC)}, pages 59--68, 1986.

\bibitem{Hof1}
C.~Hoffmann.
\newblock {\em Group-theoretic algorithms and {G}raph {I}somorphism}, volume
  136 of {\em Lecture Notes in Computer Science}.
\newblock Springer Verlag, 1982.

\bibitem{Hof2}
C.~Hoffmann.
\newblock Subcomplete generalizations of {G}raph {I}somorphism.
\newblock {\em Journal of Computer and System Sciences}, 25:332--359, 1982.

\bibitem{KST2}
J.~K\"{o}bler, U.~Sch\"{o}ning, and J.~Tor\'{a}n.
\newblock {G}raph {I}somorphism is low for {PP}.
\newblock In {\em Symposium on Theoretical Aspects of Computer Science},
  volume 577 of {\em Lecture Notes in Computer Science}, pages
  401--411. 
\newblock Springer Verlag, 1992.

\bibitem{Luk}
E.~M.~Luks.
\newblock Isomorphism of graphs of bounded valence can be tested in polynomial
  time.
\newblock {\em Journal of Computer and System Sciences}, 25:42--65, 1982.

\bibitem{Luk2}
E.~M.~Luks.
\newblock
{\it A result cited in:} L.~Babai.
Automorphism groups, isomorphism, reconstruction.
\newblock
{\em Handbook of Combinatorics, Ch.~27}, pages 1447--1540.
  Elsevier Publ., 1995.

\bibitem{Oka}
T.~Okamoto.
\newblock On relationships between statistical zero-knowledge proofs.
\newblock In {\em Proc. of the {\rm 28}th ACM Ann. Symp. on Theory of Computing
  (STOC)}, pages 649--658, 1996.

\bibitem{Scho}
U.~Sch\"{o}ning.
\newblock Graph isomorphism is in the low hierarchy.
\newblock In {\em Proceedings of the STACS'87, Lecture Notes in Computer
  Science, {\rm 247}}, pages 114--124, New York/Berlin, 1987. Springer-Verlag.

\bibitem{Sim}
C.~C.~Sims.
\newblock {\em Some group theoretic algorithms}, volume 697 of {\em Lecture
  Notes in Computer Science}, pages 108--124.
\newblock Springer Verlag, Berlin, 1978.

\end{thebibliography}
\end{document}